\documentclass[12pt]{amsart}
\usepackage[utf8]{inputenc}
\usepackage[english]{babel}
\usepackage{amsmath,amssymb,amsfonts,amsthm}
\usepackage{amscd} 
\usepackage{mathtools} 
\usepackage{longtable}
\usepackage[noadjust]{cite} 
\usepackage{enumitem}
\usepackage{float}
\usepackage[mathcal]{euscript} 
\usepackage[unicode]{hyperref}
\usepackage{ragged2e}
\usepackage{xfrac} 

\usepackage{graphicx, array}
\graphicspath{{figures/}}
\usepackage[labelfont=small,labelformat=simple]{subcaption}

\usepackage{xcolor}

\theoremstyle{plain}
\newtheorem{theorem}{Theorem}[section]
\newtheorem{lemma}[theorem]{Lemma}

\theoremstyle{definition}
\newtheorem{definition}[theorem]{Definition}

\newtheorem{remark}[theorem]{Remark}

\begin{document}

\renewcommand{\datename}{}

\renewcommand{\abstractname}{Abstract}
\renewcommand{\refname}{References}
\renewcommand{\tablename}{Table}
\renewcommand{\figurename}{Figure}
\renewcommand{\proofname}{Proof}

\title[Monodromy and Morse theory]{Hamiltonian monodromy and Morse theory}

\author{N. Martynchuk$^{1,2}$}

\author{H.W. Broer$^1$}

\author{K. Efstathiou$^1$}

\email{n.martynchuk@rug.nl}

\email{h.w.broer@rug.nl}

\email{k.efstathiou@rug.nl}

\thanks{\hspace{-5mm}$^1$ Bernoulli Institute, 
  University of Groningen, P.O. Box 407, 9700 AK Groningen,  The Netherlands}
\thanks{\hspace{-5mm}$^2$ \textit{Present address}: Department Mathematik, FAU Erlangen-N{\"u}rnberg, 
Cauerstr. 11, D-91058 Erlangen}

  \maketitle
  
  \begin{abstract}
  We show that Hamiltonian monodromy of an integrable two degrees of freedom system with a global circle action can be computed by applying Morse 
  theory to the Hamiltonian of the system. Our proof is based on Takens's index theorem, which specifies how the energy-$h$ Chern number changes when $h$ passes a non-degenerate critical value, and a 
  choice 
  of admissible cycles in
  Fomenko-Zieschang theory. Connections of our result
to some of the existing approaches to monodromy are discussed.

\mbox{ }\\
{\footnotesize The final publication is available at Springer via \\
\href{https://doi.org/10.1007/s00220-019-03578-2}{https://doi.org/10.1007/s00220-019-03578-2}
}

\end{abstract}


\section{Introduction}

Questions related to the geometry and dynamics of finite-dimensional integrable Hamiltonian systems \cite{Arnold1978, Bolsinov2004, Cushman2015} permeate modern mathematics, physics, and chemistry. They are important to such disparate fields as celestial and galactic dynamics \cite{Binney1987}, persistence and stability of invariant tori (Kolmogorov-Arnold-Moser and Nekhoroshev theories) \cite{Kolmogorov1954, Arnold1963, Moser1967, Broer1996, Salamon2004}, quantum spectra of atoms and molecules \cite{Cushman1988, Vu-Ngoc1999, Child1998, Sadovskii1999}, and the SYZ conjecture in mirror symmetry \cite{Strominger1996}. 

At the most fundamental level, a local understanding of such systems is provided by the Arnol'd-Liouville theorem \cite{Liouville1855, Arnold1968, Arnold1978, Mineur1936}.
This theorem states that integrable systems are generically foliated by tori, given by the compact and regular joint level sets of the integrals of motion, and that such foliations are always locally trivial (in the symplectic sense).
A closely related consequence of the Arnol'd-Liouville theorem, is the local existence of the \textit{action coordinates} given by the formula
$$
I_i = \int_{\alpha_i} p \, dq,
$$
where $\alpha_i, \, i = 1, \ldots, n,$ are independent homology cycles on a given torus $T^n$ of the foliation.

Passing from the local to the global description of integrable Hamiltonian systems, naturally leads to questions on the geometry of the foliation of the phase space by Arnol'd-Liouville tori. For instance, the question of whether the  bundles formed by Arnol'd-Liouville tori come from a Hamiltonian torus action, is closely connected to the existence of \textit{global} action coordinates and \emph{Hamiltonian monodromy} \cite{Duistermaat1980}. In the present work, we shall review old and discuss new  ideas related to this classical invariant.

Monodromy was introduced by Duistermaat in \cite{Duistermaat1980} and it concerns a certain `holonomy' effect that appears when one tries to construct global action coordinates for a given integrable Hamiltonian system.
If the homology cycles $\alpha_i$ appearing in the definition of the actions $I_i$ cannot be globally defined along a certain closed path in phase space, then the monodromy is non-trivial; in particular,  the system has no global action coordinates and does not admit a Hamiltonian torus action of maximal dimension (the system is not toric).

Non-trivial Hamiltonian monodromy was found in various integrable systems. The list of examples contains among others
the (quadratic) spherical pendulum
\cite{Duistermaat1980, Bates1993, Efstathiou2005, Cushman2015}, the Lagrange top \cite{Cushman1985}, 
the Hamiltonian Hopf bifurcation \cite{Duistermaat1998}, 
the champagne bottle
\cite{Bates1991}, the Jaynes-Cummings model  \cite{JaynesCummings1963, Pelayo2012, Dullin2016}, the Euler two-center and the Kepler problems \cite{Waalkens2004, DullinWaalkens2018, Martynchuk2019}. The concept of monodromy has also been extended to near-integrable systems \cite{Rink2004, Broer2007a, Broer2007}. 

In the context of monodromy and its generalizations, it is natural to ask how one can compute this invariant for a given class of integrable Hamiltonian systems. Since Duistermaat's work \cite{Duistermaat1980}, a number of different approaches to this problem, ranging from the residue calculus to algebraic and symplectic geometry, have been developed.
 The very first topological argument 
that allows one to detect non-trivial monodromy in the spherical pendulum has been given by Richard Cushman. 
Specifically, he observed that, for this system, the energy
hyper-surfaces $H^{-1}(h)$ for large values of the energy $h$ are not diffeomorphic to the energy hyper-surfaces  near the minimum where the  pendulum is at rest. 
This property is incompatible  
with the triviality of monodromy; see \cite{Duistermaat1980} and Section~\ref{sec/Morse} for more details.
This  argument demonstrates that the monodromy in the spherical pendulum is non-trivial, but does not compute it. 

Cushman's argument 
had been sleeping for many years until Floris Takens \cite{Takens2010} proposed the idea of using Chern numbers of energy hyper-surfaces 
and Morse theory for the computation of monodromy.
More specifically, he observed that  in integrable systems with a Hamiltonian circle action (such as the spherical pendulum), the Chern number of energy hyper-surfaces changes when the energy passes a critical value of the Hamiltonian function. 
The main purpose of the present paper is to explain Takens's theorem and to show that it allows one to compute monodromy in  integrable systems with a circle action.

We note that the present work is closely related to the works \cite{EfstathiouMartynchuk2017, Martynchuk2017}, which demonstrate how one can compute monodromy by focusing on the circle action and without using Morse theory. 
However,
the idea of computing monodromy through energy hyper-surfaces  and their Chern numbers can also be applied when we do not have a detailed knowledge of the singularities of the system; see Remark~\ref{remark/generalizations}. In particular, it can be applied to the case when we do not have any information about the fixed points of the circle action. We note that the behaviour of the circle action near the fixed points is important for the theory developed in the works \cite{EfstathiouMartynchuk2017, Martynchuk2017}.

The paper is organized as follows.
In Section~\ref{sec/Takens}  we discuss Takens's idea following \cite{Takens2010}.
In particular, we state and prove Takens's index theorem, which is central to the present work. In Section~\ref{sec/Morse} we show how this theorem can be applied to the context of monodromy. We discuss in detail two examples and make a connection to the Duistermaat-Heckman theorem \cite{Duistermaat1982}.
In Section~\ref{sec/connections} we revisit the symmetry approach to monodromy presented in the works
\cite{EfstathiouMartynchuk2017, Martynchuk2017}, and link it to the rotation number \cite{Cushman2015}.
The paper is concluded with a discussion in Section~\ref{sec/discussion}. 
Background material on Hamiltonian monodromy and Chern classes is presented in the Appendix.

\section{Takens's index theorem} \label{sec/Takens}
We consider an oriented $4$-manifold $M$ and a smooth Morse function $H$ on this manifold.
We recall that $H$ is called a Morse function if for
any critical (= singular) point $x$ of $H,$ the Hessian 
\begin{equation*} 
\dfrac{\partial^2 H}{\partial x_i \partial x_j}(x)
\end{equation*}
is non-degenerate.  We shall assume that $H$ is a proper \footnote{preimages of compact sets are compact} function and that 
it is invariant under a 
smooth circle action $G \colon M \times \mathbb S^1 \to M$ that is free outside the critical points of $H$. Note that the critical points of $H$ are the fixed points of the circle action.

\begin{remark} \textit{(Context of integrable Hamiltonian systems)} \label{remark/context_of_integrable} In the context of integrable systems, the function $H$ is given by the Hamiltonian of the system or another first integral, while
the circle action comes from the (rotational) symmetry. 
For instance, in the  spherical pendulum \cite{Duistermaat1980, Cushman2015}, which is a typical example of a system with monodromy,  one can take the function $H$ to be the Hamiltonian of the system; the circle action is given by the component of the angular momentum along the gravitational axis. We shall discuss this example in detail later on.
In the Jaynes-Cummings model \cite{JaynesCummings1963, Pelayo2012, Dullin2016},  one can take the function $H$ to be the integral that generates the circle action, but one can not take $H$ to be 
the Hamiltonian of the system since the latter function is not proper. 
\end{remark}

For any regular level
$H_h = \{x\in M \mid H(x) = h\},$ the circle action gives rise to the circle bundle
$$\rho_h \colon H_h \to B_h = H_h / \mathbb S^1.$$
By definition, the fibers $\rho_h^{-1}(b)$ of this bundle $\rho_h$ are the orbits of the circle action. The question that was addressed by Takens is how 
the Chern number (also known as the Euler number since it generalizes the Euler characteristic)
of this bundle changes as $h$ passes a critical value of $H$. Before stating his result  we shall make a few remarks on the Chern number 
and the circle action.

First, we note that the manifolds $H_h$ and $B_h$ are compact and admit an induced orientation. Assume, for simplicity, that $B_h$ (and hence $H_h$) are connected. Since the base manifold $B_h$ is $2$-dimensional, the (principal) circle bundle
$\rho_h \colon H_h \to B_h$ has an `almost global' section 
\begin{equation*}
s \colon  B_h   \to \rho_h^{-1}(B_h)
\end{equation*}
that is not defined at most in one point $b \in b_h.$ Let $\alpha$ be a (small) loop that encircles this point. 
\begin{definition} \label{def/Chern_number}
The \textit{Chern number} $c(h)$  of the principal bundle 
$$\rho_h \colon H_h \to B_h$$ can be defined as the winding number of $s(\alpha)$ 
along the orbit $\rho_h^{-1}(b)$. 
In other words,
  $c(h)$ is
the degree of the map
\begin{equation*} 
 S^1 = \alpha \to s(\alpha) \to \rho_h^{-1}(b) = \mathbb S^1,
\end{equation*}
where the map $s(\alpha) \to \rho_h^{-1}(b)$  is induced by a retraction of a tubular neighbourhood of $\rho_h^{-1}(b)$ onto $\rho_h^{-1}(b)$.
\end{definition}
\begin{remark}We note that the Chern number $c(h)$ is a topological invariant of the bundle $\rho_h \colon H_h \to B_h$ which
does not depend on the specific choice of the section $s$ and the loop $\alpha$; for details see \cite{Milnor1974, Postnikov1982, Fomenko2010}. 
\end{remark}

Now, consider a singular point $P$ of $H$. Observe that this point is fixed under the circle action. 
From the slice theorem \cite[Theorem I.2.1]{Audin2004} (see also \cite{Bochner1945}) it follows that in a small equivariant neighbourhood of this point the action 
can be linearized. Thus, in appropriate complex coordinates $(z,w) \in \mathbb C^2$ the action can be written
as
\begin{equation*}
  (z,w) \mapsto (e^{imt} z, e^{int} w), \ t \in \mathbb S^1,
\end{equation*}
for some integers $m$ and $n$.
By our assumption, the circle action is free outside the (isolated) critical points of the Morse function $H$.
Hence, near each such critical point the action can be written as
\begin{equation*}
  (z,w) \mapsto (e^{\pm it} z, e^{it} w), \ t \in \mathbb S^1,
\end{equation*}
in appropriate complex coordinates $(z,w) \in \mathbb C^2$.
The two cases can be mapped to each other through an orientation-reversing coordinate change. 

\begin{definition} \label{defsign}
  A singular point $P$ is called \emph{positive} if the local circle action is given by $(z,w) \mapsto (e^{-it} z, e^{it} w)$ and \emph{negative} 
  if the action is given by $(z,w) \mapsto (e^{it} z, e^{it} w)$ in a coordinate chart having the positive orientation with respect to the orientation of $M$.
\end{definition}

\begin{remark}
The Hopf fibration is defined by the circle action $(z,w) \mapsto (e^{it}z,e^{it}w)$ on the sphere 
  $$S^3 = \{(z,w) \in \mathbb C^2 \mid 1 = |z|^2+|w|^2\}.$$
  The circle action $(z,w) \mapsto (e^{-it}z,e^{it}w)$ defines the \emph{anti-Hopf fibration} on $S^3$ \cite{Urbantke2003}. If the orientation is fixed, these two fibrations are different.
\end{remark}

\begin{lemma} \label{hopfantihopf1}
 The Chern number of the Hopf fibration is equal to $-1$, while for the anti-Hopf fibration it is equal to $1$.
\end{lemma}
\begin{proof}
See Appendix~\ref{appendix/chern}.
\end{proof}

\begin{theorem} \textup{(Takens's index theorem \cite{Takens2010})} \label{theorem/Takens}
Let $H$ be a proper Morse function on an oriented $4$-manifold.
 Assume that $H$ is invariant under a circle action that is free outside the critical points. Let $h_c$ be a critical value of $H$ containing exactly one critical point. Then the Chern numbers 
 of the nearby levels satisfy
 $$c(h_c+\varepsilon) = c(h_c-\varepsilon) 
 \pm 1.
 $$
 Here the sign is plus if the circle action defines the anti-Hopf fibration near the critical point and minus for the Hopf fibration. 
 
\end{theorem}

\begin{proof}
The main idea is to apply Morse theory to the function $H$. The role of Euler characteristic in standard Morse theory will be played by the Chern number. We note that the Chern number, just like 
the Euler characteristic, is additive.

From Morse theory \cite{Milnor1963}, 
we have that the manifold $H^{-1}(-\infty,h_c+\varepsilon]$ can be obtained from the manifold
$H^{-1}(-\infty,h_c-\varepsilon]$ by attaching a handle $D^{\lambda} \times D^{4-\lambda}$, where $\lambda$ is the index of the critical point on the level $H^{-1}(h_c).$
More specifically, for a suitable neighbourhood  $D^{\lambda} \times D^{4-\lambda} \subset M$   
of the critical point (with $D^m$ standing for an $m$-dimensional ball), $H^{-1}(-\infty,h_c+\varepsilon]$ deformation retracts onto
the set
\begin{equation*} \label{equation/proof_Takens_1}
 X = H^{-1}(-\infty,h_c-\varepsilon] \cup D^{\lambda} \times D^{4-\lambda}
\end{equation*}
and, moreover,
\begin{equation} \label{eq/Morse}
 H^{-1}(-\infty,h_c+\varepsilon] \simeq X = H^{-1}(-\infty,h_c-\varepsilon] \cup D^{\lambda} \times D^{4-\lambda}
\end{equation}
up to a diffeomorphism. We note that by the construction, the
intersection of the handle $D^{\lambda} \times D^{4-\lambda}$ with $H^{-1}(-\infty,h_c-\varepsilon]$ is the subset $S^{\lambda-1} \times D^{4-\lambda} \subset H^{-1}(h_c-\varepsilon);$ see \cite{Milnor1963}. For simplicity, we shall assume that  the handle is disjoint from $H^{-1}(h_c +\varepsilon)$. By taking the boundary  in Eq~\eqref{eq/Morse}, we get that
\begin{equation} \label{equation/proof_Takens}
 H^{-1}(h_c+\varepsilon) \simeq \partial X =  (H^{-1}(h_c-\varepsilon) \setminus S^{\lambda-1} \times D^{4-\lambda}) \cup D^{\lambda} \times S^{4-\lambda-1}.
\end{equation}
Here the union $(D^{\lambda} \times S^{4-\lambda-1}) \cup (S^{\lambda-1} \times D^{4-\lambda})$ is the boundary $S^3 = \partial(D^{\lambda} \times D^{4-\lambda})$ of the handle.

Since we assumed the existence of a global circle action on $M$, 
we can choose the handle and its boundary $S^3$ to be invariant with respect to this action \cite{Wasserman1969}. 
This will allow us 
to relate the  Chern numbers of $H^{-1}(h_c+\varepsilon)$ and $H^{-1}(h_c-\varepsilon)$ using 
Eq.~\eqref{equation/proof_Takens}. 
Specifically, 
due to the invariance under the circle action, the sphere  
$S^3$
has a well-defined Chern number. Moreover, since the action is assumed to be free outside the critical points of $H$, this Chern number
$c(S^3) = \pm 1$, depending on whether the circle action defines the anti-Hopf or the Hopf fibration on $S^3$; see Lemma~\ref{hopfantihopf1}.
From Eq.\eqref{equation/proof_Takens} and the additive property of the Chern number, we get
$$c(\partial X) = c(h_c-\varepsilon) + c(S^3) = c(h_c-\varepsilon) \pm 1.$$
It is left to show that $c(h_c+\varepsilon) = c(\partial X)$ (we note that even though we know that $H^{-1}(h_c+\varepsilon)$ and $\partial X$ are diffeomorphic, we cannot yet conclude that they have the same Chern numbers). 

Let the subset $ Y \subset M$ be defined as the closure of the set
$$
H^{-1}[h_c -\varepsilon, h_c +\varepsilon] \setminus D^{\lambda} \times D^{4-\lambda}.
$$
We observe that $Y$ is a compact submanifold of $M$ and that $\partial Y = \partial X \cup H^{-1}(h_c+\varepsilon),$ that is, $Y$ is a cobordism in $M$ between $\partial X$ and $H^{-1}(h_c+\varepsilon)$. 
By the construction,
$\partial Y$ is invariant under the circle action and there are no critical points of $H$ in $Y$.
It follows that the Chern number  $c(\partial Y) = 0$. Indeed, one can apply Stokes's theorem to the Chern class of $\rho \colon Y \to Y / \mathbb S^1$, where $\rho$ is the reduction map; see Appendix~\ref{appendix/chern}. This concludes the proof of the theorem.
\end{proof}

\begin{remark}
We note that (an analogue of) Theorem~\ref{theorem/Takens} holds also when the Hamiltonian function $H$ has $k > 1$ isolated critical points on a critical  level. In this case 
$$c(h_c+\varepsilon) = c(h_c-\varepsilon) 
 + \sum\limits_{i = 1}^k s_k,
 $$
 where $s_k = \pm 1$ corresponds to the $k$th critical point.
\end{remark}

\begin{remark}
By a continuity argument, the (integer) Chern number is locally constant. This means that if $[a,b]$ does not contain critical values of $H$, then $c(h)$ is the same for all
the values $h \in [a,b]$.
On the other hand, by
Theorem~\ref{theorem/Takens}, the Chern number $c(h)$ changes when $h$ passes a critical value which corresponds to a single critical point.
\end{remark}

\section{Morse theory approach to monodromy} \label{sec/Morse}

The goal of the present section is to show how Takens's index theorem can be used to compute Hamiltonian monodromy. First, we demonstrate our method 
on a famous example of a system with non-trivial monodromy: the
\textit{spherical pendulum}. Then, we give a new proof of the \textit{geometric monodromy theorem} along similar lines. We also show that the jump in the energy level Chern number manifests 
non-triviality of Hamiltonian monodromy in the general case. This section is concluded with studying Hamiltonian monodromy in an example of an integrable system with two focus-focus points.

\subsection{Spherical pendulum} \label{subsec/spherical_pendulum}
The spherical pendulum
describes the motion of a particle moving on the unit sphere 
\begin{align*}
S^2 = \{ (x,y,z) \in \mathbb R^3 \colon x^2 + y^2 + z^2 = 1 \}
\end{align*}
in the linear gravitational potential $V(x,y,z) = z.$
The corresponding Hamiltonian system is given by $$(T^{*}S^2, \Omega|_{T^{*}S^2}, H|_{T^{*}S^2}),  \mbox{ where } H = \frac{1}{2}(p_x^2+p_y^2+p_z^2) + V(x,y,z)$$ is the total energy of the
pendulum and $\Omega$ is the standard symplectic structure. We observe that the function $J = xp_y - yp_x$ (the component of the total angular momentum about the $z$-axis) is conserved. It follows 
that the system is Liouville integrable.
The \textit{bifurcation diagram} of the  energy-momentum map 
$$F = (H,J) \colon T^{*}S^2 \to \mathbb R^2,$$ that is, the set of the critical values of this map, is shown in Fig.~\ref{fig/ct1}.

\begin{figure}[ht]
\begin{center}
\includegraphics[width=0.95\linewidth]{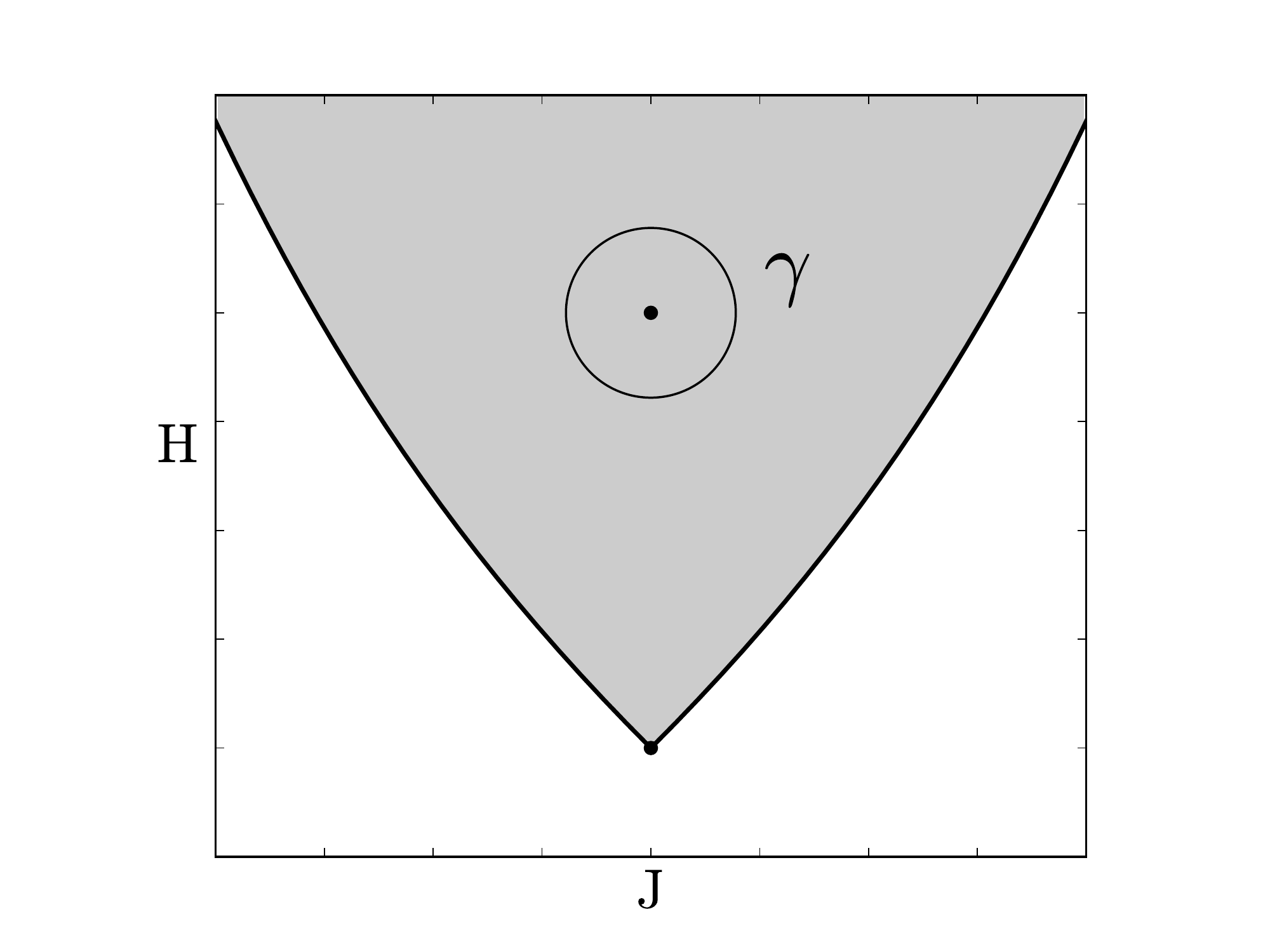}
\caption{Bifurcation diagram for the spherical pendulum and the loop $\gamma$ around the focus-focus singularity. }
\label{fig/ct1}
\end{center}
\end{figure}

From the bifurcation diagram we see that the set $R \subset \textup{image}(F)$ of the regular values of $F$ (the shaded area in Fig.~\ref{fig/ct})
is an open subset of $\mathbb R^2$ with one puncture. Topologically, $R$ is an annulus and hence 
$\pi_1(R, f_0) = \mathbb Z$ for any $f_0 \in R$. We note that
the puncture (the black dot in Fig.~\ref{fig/ct1}) corresponds to an isolated singularity; specifically, 
to the unstable equilibrium of the pendulum. 

Consider the closed path $\gamma$ around the puncture that is shown in Fig.~\ref{fig/ct1}. Since $J$ generates a Hamiltonian circle action on $T^{*}S^2$, any orbit of this action
on $F^{-1}(\gamma(0))$ can be transported along $\gamma$. Let $(a,b)$ be a basis of $H_1(F^{-1}(\gamma(0)))$, where $b$ is given by the homology class of such an orbit. Then the corresponding Hamiltonian
monodromy matrix along $\gamma$ is given by
\begin{equation*}
 M_\gamma = \begin{pmatrix} 1 & m_\gamma \\ 0 & 1\end{pmatrix}
 \end{equation*}
 for some integer $m_\gamma$. It was shown in \cite{Duistermaat1980} that $m_\gamma = 1$ (in particular, global action coordinates do not exist in this case).
Below we shall show how this result follows from Theorem~\ref{theorem/Takens}.

First we recall the following argument due to Cushman, which
shows that the monodromy along the loop $\gamma$ is non-trivial; the argument appeared in \cite{Duistermaat1980}.

{\bfseries Cushman's argument.}
First
observe that the points 
$$P_{min} = \{p = 0, z = -1\} \ \mbox{ and } \ P_{c} = \{p = 0, z = 1\}$$ are the only critical points of $H$. The corresponding critical values 
are $h_{min} = -1$ and $h_c = 1$, respectively. The point $P_{min}$ is the global and non-degenerate minimum of $H$ on $T^{*}S^2$.
From the Morse lemma, we have that $H^{-1}(1 - \varepsilon), \ \varepsilon \in (0,2),$ is diffeomorphic to the $3$-sphere $S^3$. On the other hand,  $H^{-1}(1 + \varepsilon)$
is diffeomorphic to the unit cotangent bundle $T^{*}_1S^2$. It follows that the monodromy index $m_\gamma \ne 0$.
Indeed, the energy levels $H^{-1}(1 + \varepsilon)$ and $H^{-1}(1 - \varepsilon)$ are isotopic, respectively, to  $F^{-1}(\gamma_1)$ and $F^{-1}(\gamma_2)$,
where $\gamma_1$ and $\gamma_2$ are the curves shown in Fig.~\ref{fig/ct}. 
If $m_\gamma = 0$, then the preimages $F^{-1}(\gamma_1)$ and $F^{-1}(\gamma_2)$  would be homeomorphic, 
which is not the case. \hfill $\square$

\begin{figure}[ht]
\begin{center}
\includegraphics[width=0.98\linewidth]{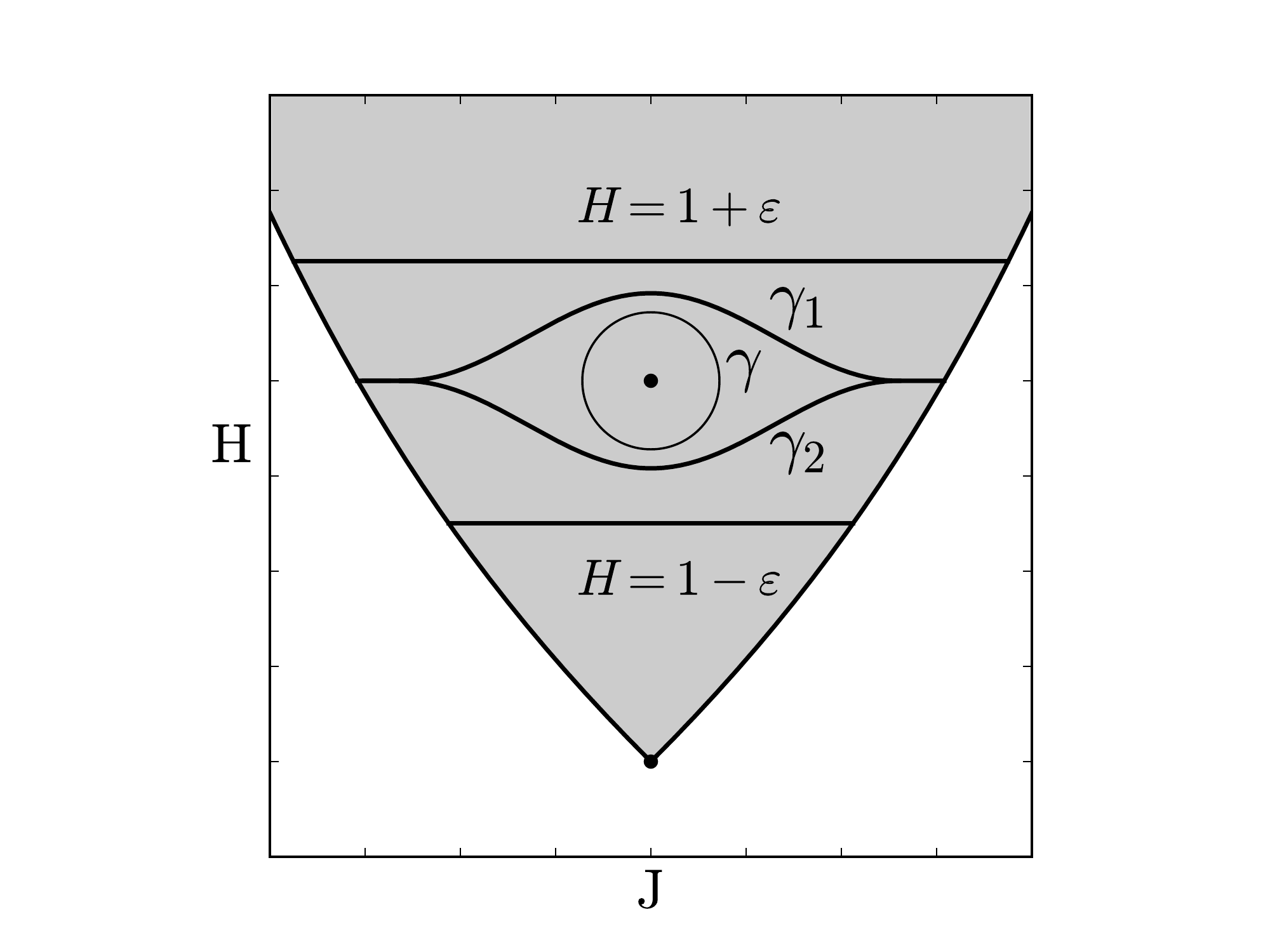}
\caption{Bifurcation diagram for the spherical pendulum, the energy levels, the curves $\gamma_1$ and $\gamma_2$, and the loop $\gamma$ around the focus-focus singularity. }
\label{fig/ct}
\end{center}
\end{figure}

Using Takens's index theorem~\ref{theorem/Takens}, we shall now make one step further and compute the monodromy index $m_\gamma$. By Takens's index theorem, the energy-level Chern numbers are related via
$$c(1+\varepsilon) = c(1-\varepsilon) 
 + 1
 $$
 since the critical point $P_c$ is of \textit{focus-focus} type. Note that focus-focus points are positive by Theorem~\ref{zungaction1}; for a definition of focus-focus points
 we refer to \cite{Bolsinov2004}.  
 
 Consider again the curves $\gamma_1$ and $\gamma_2$ shown in Fig.~\ref{fig/ct}. Observe that $F^{-1}(\gamma_1)$ and $F^{-1}(\gamma_2)$ are invariant under the circle action 
 given by the Hamiltonian flow of $J$. Let $c_1$ and $c_2$ denote the corresponding Chern numbers. By the isotopy, we have that 
 $c_1 = c(1+\varepsilon)$ and $c_2 = c(1-\varepsilon).$ 
 In particular, $c_1 =  c_2 + 1$.
 
 Let $\delta > 0$ be sufficiently small. Consider the following set 
$$S_{-} = \{x \in F^{-1}(\gamma_1) \mid J(x) \le j_{min}+\delta\},$$ 
where $j_{min} $ is the minimum value of the momentum $J$ on $F^{-1}(\gamma_1)$. Similarly, we define the set
$$S_{+} = \{x \in F^{-1}(\gamma_1) \mid J(x) \ge j_{max}-\delta\}.$$ 
By the construction of the curves $\gamma_i$, the sets $S_{-}$ and $S_{+}$ are contained in both $F^{-1}(\gamma_1)$ and $F^{-1}(\gamma_2)$. 
Topologically, these sets are solid tori.

Let $(a_{-},b_{-})$ be two basis cycles on $\partial S_{-}$ such that $a_{-}$ is the meridian and $b_{-}$ is an orbit of the circle action. Let 
$(a_{+},b_{+})$ be the corresponding cycles on $\partial S_{+}$. The preimage $F^{-1}(\gamma_i)$ is homeomorphic to the space obtained by gluing these pairs of cycles by
$$
\begin{pmatrix}
 a_{-}  \\
 b_{-} 
\end{pmatrix} = \begin{pmatrix}
 1 & c_i \\
 0 & 1
\end{pmatrix} \begin{pmatrix}
 a_{+} \\
 b_{+}
\end{pmatrix},
$$
where $c_i$ is the Chern number of $F^{-1}(\gamma_i)$. It follows that the monodromy matrix along $\gamma$ is given by the  product
$$
M_\gamma = \begin{pmatrix}
 1 & c_1 \\
 0 & 1
\end{pmatrix} \begin{pmatrix}
 1 & c_2 \\
 0 & 1
\end{pmatrix}^{-1}.
$$
Since $c_1 = c_2 
 + 1,$ we conclude that the monodromy matrix
$$
M_\gamma = \begin{pmatrix}
 1 & 1 \\
 0 & 1
\end{pmatrix}.
$$

\begin{remark} \textit{(Fomenko-Zieschang theory)} \label{remark/Fomenko_Zieschang_Hamiltonian_monodromy}
The cycles $a_{\pm}, b_{\pm}$, which we have used when expressing $F^{-1}(\gamma_i)$ as a result of gluing two solid tori, are \textit{admissible} in the sense of Fomenko-Zieschang theory 
\cite{Fomenko1990, Bolsinov2004}. It follows, in particular, that the Liouville fibration of $F^{-1}(\gamma_i)$ is determined by the \textit{Fomenko-Zieschang invariant} (the \textit{marked molecule})
$$
A^* \raisebox{.5ex}{\rule{2.1cm}{.20pt}} \mbox{ }^{\hspace{-2.04cm}r_i = \infty, \hspace{1mm} \varepsilon = 1, \hspace{1mm} n_i} \  A^*
$$
with the $n$-\textit{mark} $n_i$ given by the Chern number $c_i$. (The same is true for the regular energy levels $H^{-1}(h)$.) Therefore, our results show that Hamiltonian monodromy 
is also given by the jump in the  $n$-mark. We note that the $n$-mark and the other 
labels in the Fomenko-Zieschang invariant are also defined in the case when no global circle action exists.
\end{remark}

\subsection{Geometric monodromy theorem} \label{subsec/geometric_monodromy_theorem}
A common aspect of most of the systems with non-trivial Hamiltonian monodromy is that the corresponding 
energy-momentum map has focus-focus points, which, from the perspective of Morse theory, are saddle points of the Hamiltonian function.

The following result, which is sometimes referred to as the \textit{geometric monodromy theorem}, characterizes monodromy around a focus-focus singularity in systems with two degrees of freedom.

\begin{theorem} \textup{(Geometric monodromy theorem, \cite{Matsumoto1989, Lerman1994, Matveev1996, Zung1997})} \label{theorem/geometric_monodromy_theorem0}
Monodromy around a focus-focus singularity is given
 by the matrix
 \begin{equation*}
 M = \begin{pmatrix} 1 & m \\ 0 & 1\end{pmatrix},
 \end{equation*}
 where $m$ is the number of the focus-focus points on the singular fiber.
\end{theorem}

 A related  result in the context of the focus-focus singularities is that they come  with a Hamiltonian circle action \cite{Zung1997, Zung2002}.

\begin{theorem}\textup{(Circle action near focus-focus, \cite{Zung1997, Zung2002})} \label{zungaction1} 
In a neighbourhood of a focus-focus fiber \footnote{that is, a singular fiber containing a number of focus-focus points.}, there exists a unique (up to orientation reversing) Hamiltonian 
circle action which is free everywhere except for the singular focus-focus points.
Near each singular point, the momentum of the circle action can be written as
\begin{align*}
  J = \frac12 (q_1^2+p_1^2) - \frac12 (q_2^2+p_2^2)
\end{align*}
for some local canonical coordinates $(q_1, p_1, q_2, p_2)$. In particular, the circle action defines the anti-Hopf fibration near each singular point.
\end{theorem}

 One implication of Theorem~\ref{zungaction1} is that it allows to prove the geometric monodromy theorem by looking at the circle action. 
 Specifically, one can apply the Duistermaat-Heckman theorem in this case;
 see
 \cite{Zung2002}. 
A slight modification of our argument, used in the previous Subsection~\ref{subsec/spherical_pendulum} to determine monodromy in the spherical pendulum,
results in another proof of the geometric monodromy theorem. We give this proof below. 

\renewcommand*{\proofname}{Proof of Theorem~\ref{theorem/geometric_monodromy_theorem0}}

\begin{proof}
By applying integrable surgery, we can assume that the bifurcation diagram consists of a square of elliptic singularities and a focus-focus singularity in the middle; see \cite{Zung2002}.
In the case when there is only one focus-focus point on the singular focus-focus fiber,
the proof reduces to the case of the spherical pendulum. Otherwise the configuration is unstable. Instead of a focus-focus fiber with $m$ singular points, one can consider 
a new  $\mathbb S^1$-invariant fibration such that it is arbitrary close to the original one and has $m$ simple (that is, containing only one critical point) focus-focus fibers; see 
Fig.~\ref{Figure1_GMTproof}. 
\begin{figure}[ht]
\begin{center}
\includegraphics[width=0.75\linewidth]{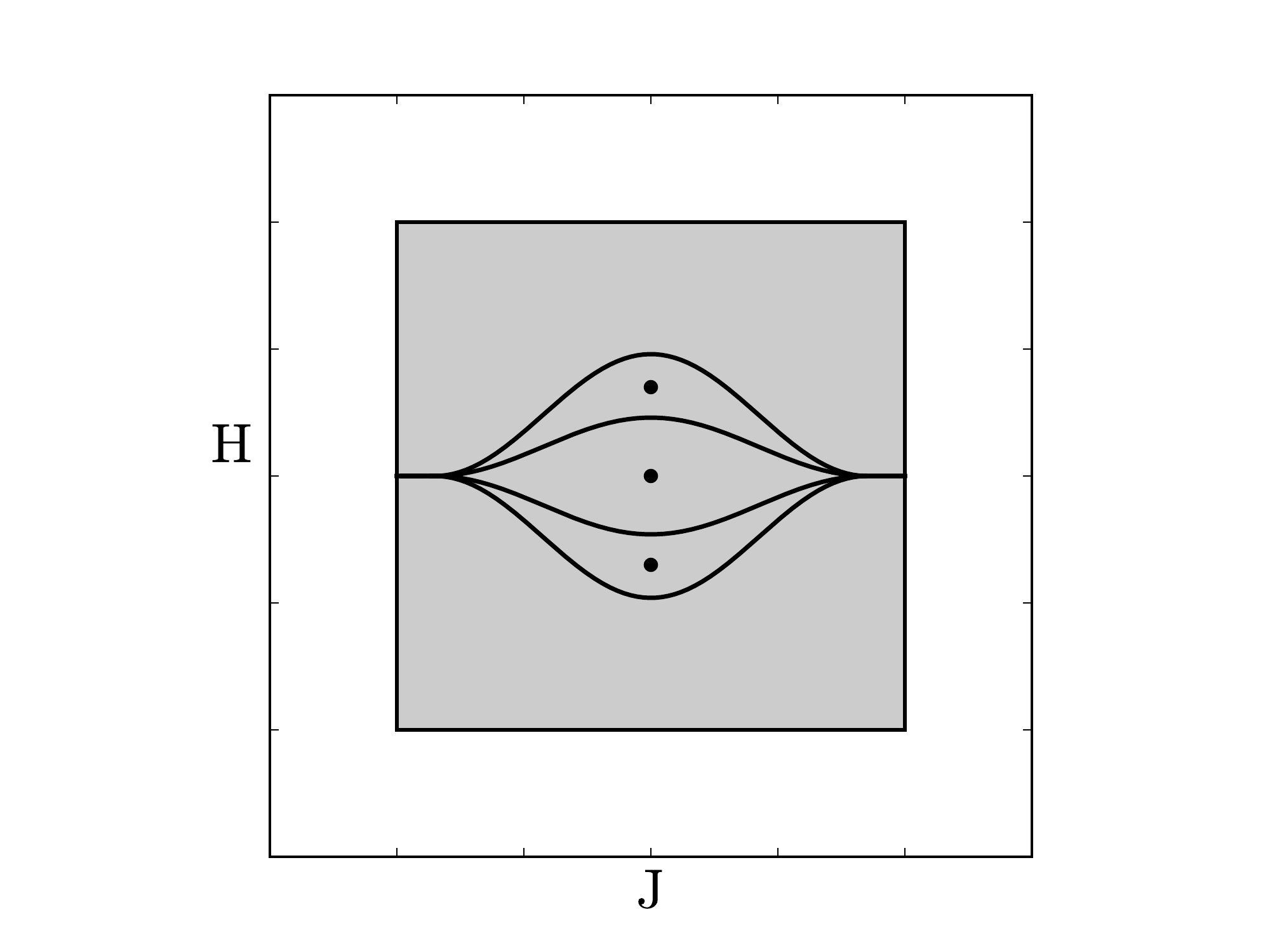}
\caption{Splitting of the focus-focus singularity; the complexity $m = 3$ in this example.}
\label{Figure1_GMTproof}
\end{center}
\end{figure}
As in the case of the spherical pendulum, we get that the monodromy matrix around each of the simple focus-focus fibers is given by the matrix 
$$
M_i = \begin{pmatrix}
 1 & 1 \\
 0 & 1
\end{pmatrix}.
$$
Since the new fibration is $\mathbb S^1$-invariant, the monodromy matrix around $m$ focus-focus fibers is given by the product of $m$ such matrices, that is,
$$
M_\gamma = M_1 \cdots M_m  = \begin{pmatrix}
 1 & m \\
 0 & 1
\end{pmatrix}.
$$
The result follows.
\end{proof}

\renewcommand*{\proofname}{Proof}

\begin{remark} \textit{(Duistermaat-Heckman)} Consider a symplectic $4$-manifold $M$ and a proper function $J$ that generates
a Hamiltonian  circle action on this manifold. Assume that the fixed points are isolated and that the action is free outside these points.
From the Duistermaat-Heckman theorem \cite{Duistermaat1982} it follows that 
 the symplectic volume $vol(j)$ of $J^{-1}(j)/\mathbb S^1$ is a piecewise linear function. Moreover, if $j = 0$ is a critical value with $m$ positive fixed points of the circle action, then
 $$ vol(j) + vol(-j) = 2vol(0)  - m j.$$
As was shown in \cite{Zung2002}, this result implies the geometric monodromy theorem since the symplectic volume can be viewed as the affine length of the line segment $\{J = j\}$ in the image of $F$.
The connection to our approach can be seen from the observation that the derivative $vol'(j)$ coincides with the Chern number of  $J^{-1}(j)$. 
We note that for the spherical pendulum,
the Hamiltonian does not generate a circle action, whereas
the $z$-component of the angular momentum is not a proper function. Therefore, neither of these functions can be taken as `$J$'; in order to use the Duistermaat-Heckman theorem, one 
needs to consider a local model first \cite{Zung2002}. 
Our approach, based on Morse theory, can be applied directly to the Hamiltonian of the spherical pendulum, even though it does not generate a circle action.
\end{remark}

\begin{remark} \textit{(Generalization)}\label{remark/generalizations}
We observe that even if a simple closed curve $\gamma \subset R$ bounds some complicated arrangement of singularities or, more generally, if the interior of $\gamma$ in $\mathbb R^2$ is not contained in the image of  the energy-momentum map $F$,  the 
monodromy along this curve can still be computed by looking at the energy level Chern numbers. 
Specifically, the  monodromy along $\gamma$ is given by
$$
M_\gamma = \begin{pmatrix}
 1 & m_\gamma \\
 0 & 1
\end{pmatrix},
$$
where $m_\gamma = c(h_2) - c(h_1)$ is the difference between the Chern numbers of two (appropriately chosen) energy levels.
\end{remark}

\begin{remark} \textit{(Planar scattering)}
We note that a similar result holds 
in the case of mechanical Hamiltonian systems on $T^{*}\mathbb R^2$ that are both scattering and integrable; see \cite{Martynchuk2016}. 
For such systems, the roles of the compact monodromy and the Chern number
are played by the \textit{scattering monodromy} and \textit{Knauf's scattering index} \cite{Knauf1999}, respectively. 
\end{remark}

\begin{remark} \textit{(Many degrees of freedom)}
 The approach presented in this paper depends on the use of energy-levels and their Chern numbers. For this reason, it cannot be directly generalized to systems with many degrees of freedom. 
 An approach that admits such a generalization was developed in \cite{EfstathiouMartynchuk2017, Martynchuk2017}; we shall recall it in the next section.
\end{remark}

\subsection{Example: a system with two focus-focus points}

Here we illustrate the Morse theory approach that we developed in this paper on a concrete example of an integrable system that has more than one focus-focus point. The system was introduced in
 \cite{HohlochPalmer2018}; it is
 an example 
of a 
\textit{semi-toric system} \cite{Vu-Ngoc2007, Sepe2014, DullinPelayo2016} with a special property that it has two distinct focus-focus fibers, which are not on the same 
level of the momentum corresponding to the circle action.

Let $S^2$ be the unit sphere in $\mathbb R^3$ and let $\omega$ denote its volume form, induced from $\mathbb R^3$. Take the product $S^2 \times S^2$ with the symplectic structure $\omega \oplus 2\omega$. 
The system introduced in \cite{HohlochPalmer2018} is an integrable system 
on $S^2 \times S^2$ defined in Cartesian coordinates $(x_1,y_1,z_1,x_2,y_2,z_2) \in \mathbb R^3 \oplus \mathbb R^3$ by the Poisson commuting functions
\begin{equation*}
 H = \frac{1}{4}z_1+\frac{1}{4}z_2 + \frac{1}{2}(x_1x_2+y_1y_2) \ \ \mbox{ and } \ \
 J = z_2 +2z_2.
\end{equation*}
The bifurcation diagram of the corresponding energy-momentum map 
$F = (H,J) \colon S^2\times S^2 \to \mathbb R^2$ is shown in Fig.~\ref{Figure/2focusfocus}.

\begin{figure}[ht]
\begin{center}
\includegraphics[width=0.9\linewidth]{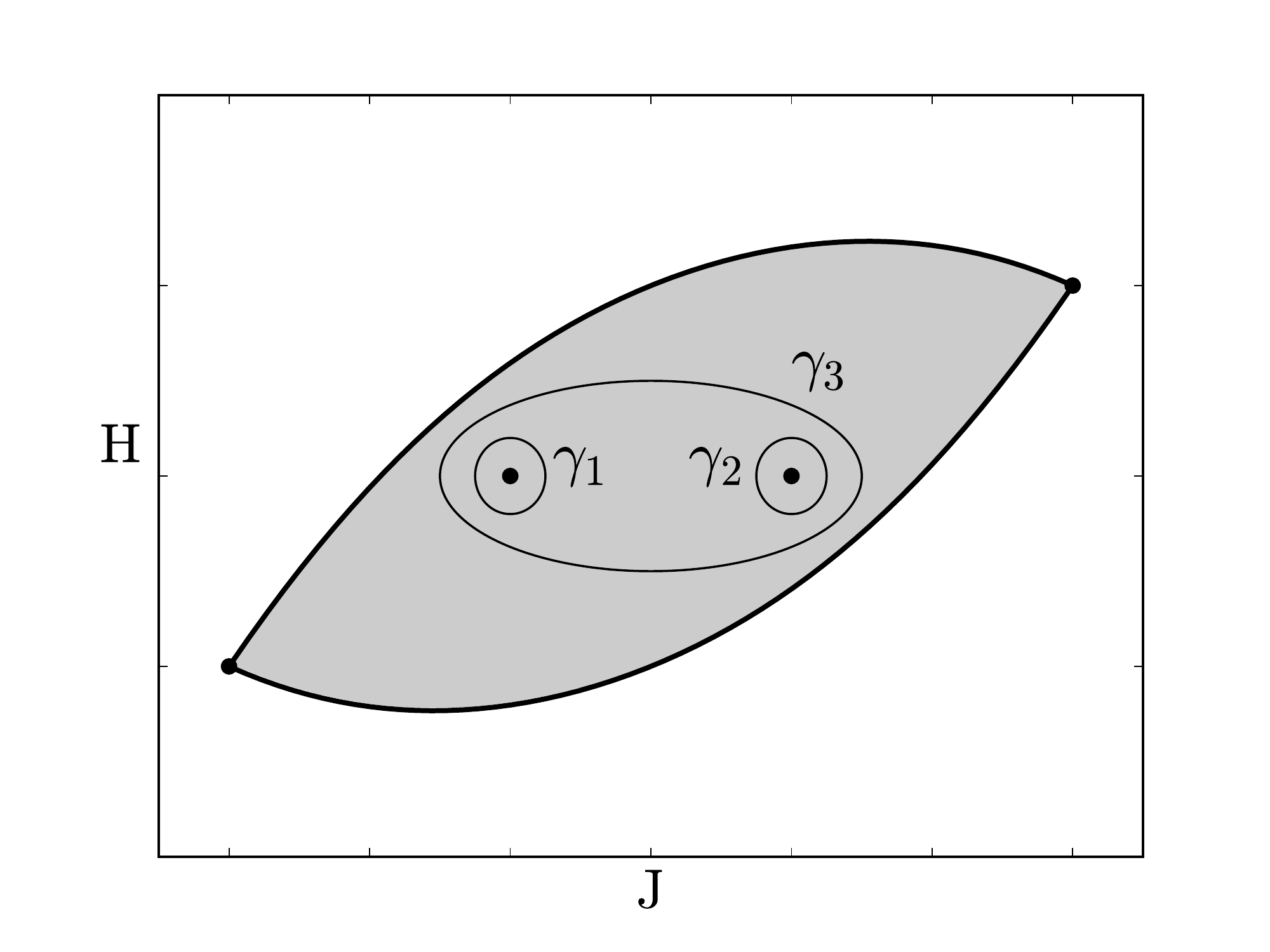}
\caption{The bifurcation diagram for the system on $S^2\times S^2$ and the loops $\gamma_1, \gamma_2, \gamma_3$ around the focus-focus singularities.}
\label{Figure/2focusfocus}
\end{center}
\end{figure}

The system has 4 singular points: two focus-focus and two elliptic-elliptic points. These singular points are $(S,S), (N,S), (S,N)$ and $(N,N)$, where $S$ and $N$ are the 
South and the North poles of $S^2$. Observe that these points are the fixed points of the circle action generated by the momentum $J$. The focus-focus points are positive fixed points 
(in the sense of Definition~\ref{defsign}) and the elliptic-elliptic points are negative. Takens's index theorem implies that the topology of the regular $J$-levels are
$S^3, S^2\times S^1,$ and $S^3$; the corresponding Chern numbers are $-1,0,$ and $1$, respectively. Invoking the argument in Subsection~\ref{subsec/spherical_pendulum} for the spherical pendulum 
(see also Subsection~\ref{subsec/geometric_monodromy_theorem}), we conclude\footnote{We note that Eq.~\ref{eq/monodromytwofocus} follows also from the geometric monodromy theorem since the circle action gives a universal sign for the monodromy around the two focus-focus points \cite{Cushman2002}. Our aim is to prove Eq.~\ref{eq/monodromytwofocus} by looking at the topology of the energy levels.} that the monodromy matrices along the curves $\gamma_1$ and $\gamma_2$ that encircle 
the focus-focus points (see Fig.~\ref{Figure/2focusfocus})
are
\begin{equation} \label{eq/monodromytwofocus}
M_1 = M_2 = \begin{pmatrix}
 1 & 1 \\
 0 & 1
\end{pmatrix}.
\end{equation}
Here the homology basis $(a,b)$ is chosen such that $b$ is an orbit of the circle action. 

\begin{remark}
Observe that the regular $H$-levels have the following topology: $S^2\times S^1, S^3, S^3,$ and $S^2\times S^1$. We see that the energy levels do not change their topology as the value 
of $H$ passes the critical value $0$, which corresponds to the two focus-focus points. Still, the monodromy around $\gamma_3$ is nontrivial. Indeed, in view of Eq.~\eqref{eq/monodromytwofocus} and the existence of 
a global circle action \cite{Cushman2002}, the monodromy along $\gamma_3$ is given by
\begin{equation*}
M_3 = M_1 \cdot M_2 = \begin{pmatrix}
 1 & 2 \\
 0 & 1
\end{pmatrix}.
\end{equation*}
The apparent paradox is resolved when one looks at the Chern numbers: the Chern number of the $3$-sphere below the focus-focus points is equal to $-1$, whereas 
the Chern number of the $3$-sphere above the focus-focus points is equal to $+1$. (The Chern number of   $S^2\times S^1$ is equal to $0$ in both cases.) We note that a similar kind of example of an 
integrable system for which  the monodromy is non-trivial and the energy levels do not change their topology, is given in \cite{Cushman2015} (see Burke's egg (poached)). In the case of
Burke's egg, the energy levels are non-compact; in the case of the system on $S^2\times S^2$ they are compact.
\end{remark}

\section{Symmetry approach} \label{sec/connections}

We note that one can avoid using energy levels by looking directly at the Chern number of $F^{-1}(\gamma)$, where $\gamma$ is the closed curve along which Hamiltonian monodromy is defined. 
This point of view was developed in the work \cite{EfstathiouMartynchuk2017}. It is based on the following two results.

\begin{theorem} \textup{(Fomenko-Zieschang, \cite[\S 4.3.2]{Bolsinov2004}, \cite{EfstathiouMartynchuk2017})} \label{theorem/Bolsinov_Fomenko_Zieschang}
Assume that the energy-momentum map $F$ is proper and invariant under a Hamiltonian circle action. 
Let $\gamma \subset \textup{image}(F)$ be  a simple closed curve in the set of the regular values of the map $F$.  Then
 the Hamiltonian
monodromy of the torus bundle $F \colon F^{-1}(\gamma) \to \gamma$ is given by
$$
\begin{pmatrix}
 1 & m \\
 0 & 1
\end{pmatrix} \in \mathrm{SL}(2,\mathbb Z),
$$
where $m$ is the Chern number of the principal circle bundle $\rho \colon F^{-1}(\gamma) \to  F^{-1}(\gamma) / \mathbb S^1$, defined by reducing the circle action.
\end{theorem}

In the case when the curve $\gamma$ bounds a disk $D \subset \textup{image}(F)$,  the Chern number $m$ can be computed from the singularities of the circle action that project into $D$. Specifically,
there is the following result.

 \begin{theorem} \textup{(\cite{EfstathiouMartynchuk2017})}  \label{theorem/EM}
Let $F$ and $\gamma$ be as in Theorem~\ref{theorem/Bolsinov_Fomenko_Zieschang}. Assume that $\gamma = \partial D$, where $D \subset \textup{image}(F)$ is a two-disk, and that
the circle action is free everywhere in $F^{-1}(D)$ outside isolated fixed points. Then 
the Hamiltonian monodromy of the $2$-torus bundle $F \colon F^{-1}(\gamma) \to \gamma$ is given by 
the number of positive singular points minus the number of negative singular points in $F^{-1}(D)$.
\end{theorem}

 We note that Theorems~\ref{theorem/Bolsinov_Fomenko_Zieschang} and \ref{theorem/EM} were generalized to a much more general setting of fractional monodromy and Seifert fibrations; see \cite{Martynchuk2017}. 
 Such a generalization allows one, in particular, to define monodromy for circle bundles over 2-dimensional surfaces (or even orbifolds) of genus $g \ge 1$; in the standard case the genus  $g = 1$.

Let us now give a new proof of Theorem~\ref{theorem/Bolsinov_Fomenko_Zieschang}, which makes a connection to the rotation number. First we shall recall this notion.

We assume that the energy-momentum map $F$ is invariant under a Hamiltonian circle action. Without loss of generality, $F = (H,J)$ is such that
the circle action is given by the Hamiltonian flow $\varphi^t_J$  of $J$. 
Let $F^{-1}(f)$ be a regular torus. Consider a point $x \in F^{-1}(f)$ and the orbit of the circle action passing through this point. 
The trajectory $\varphi^t_H(x)$ leaves the orbit of the circle action at $t = 0$ and then returns back to the same orbit at some time $T > 0$. The time $T$ is called the \textit{the first return time}.
The \textit{rotation number} $\Theta = \Theta(f)$ is defined by $\varphi^{2\pi {\Theta}}_J(x) = \varphi^{T}_H(x)$.  There is the following result.

\begin{theorem} \textup{(Monodromy and rotation number, \cite{Cushman2015})} \label{theorem/rotation_number_mon}
 The Hamiltonian monodromy of the torus bundle $F \colon F^{-1}(\gamma) \to \gamma$ is given by
$$
\begin{pmatrix}
 1 & m \\
 0 & 1
\end{pmatrix} \in \mathrm{SL}(2,\mathbb Z),
$$
where $-m$ is the variation of the rotation number $\Theta$.
\end{theorem} 
\begin{proof}
First we note that since the flow of $J$ is periodic on $F^{-1}(\gamma)$, the monodromy matrix is of the form 
$$
\begin{pmatrix}
 1 & m' \\
 0 & 1
\end{pmatrix} \in \mathrm{SL}(2,\mathbb Z)
$$
for some integer $m'$. 

Fix a starting point $f_0 \in \gamma.$ Choose a smooth branch of the rotation number $\Theta$ on $\gamma \setminus f_0$  and define the vector field $X_S$ on $F^{-1}(\gamma \setminus f_0)$ by
\begin{equation} \label{periodic_vector_field}
 X_S = \frac{T}{2\pi} X_H - \Theta X_J.
\end{equation}
 By the construction,
 the  flow of $X_S$ is periodic. However, unlike the flow of $X_J,$ it is not globally defined on $F^{-1}(\gamma).$
 Let $\alpha_1$ and $\alpha_0$ be the limiting cycles of this vector field on $F^{-1}(f_0),$ that is, let $\alpha_0$ be given by the flow of the vector field $X_S$ for $\xi \to f_0+$ and
 let $\alpha_1$ be given by the flow of $X_S$ for $\xi \to f_0-$. Then 
$$
  \alpha_1 = \alpha_0 + m b_{f_0},
$$
where $-m$ is the variation of the rotation number along $\gamma$. Indeed, if the variation of the rotation number  is $-m$, then the vector field
$\frac{T(f_0)}{2\pi} X_H - \Theta(f_0) X_J$ on $F^{-1}(f_0)$ changes to  $\frac{T(f_0)}{2\pi} X_H - (\Theta(f_0) - m) X_J$  after $\xi$ traverses  $\gamma.$
Since $\alpha_1$ is the result of the parallel transport of $\alpha_0$ along $\gamma$, we conclude that $m' = m$. The result follows.
\end{proof}

We are now ready to prove Theorem~\ref{theorem/Bolsinov_Fomenko_Zieschang}.
 
\begin{proof}

Take an invariant metric $g$ on $F^{-1}(\gamma)$ and define a connection $1$-form $\sigma$
 of the principal $\mathbb S^{1}$ bundle $\rho \colon E_\gamma \to E_\gamma / \mathbb S^{1}$ as follows: 
$$\sigma(X_J) = i \ \mbox{ and } \ \sigma(X_H) = \sigma(e) = 0,$$
where $e$ is orthogonal to $X_J$ and $X_H$ with respect to the metric $g$.
Since the flows $\varphi^t_H$ and $\varphi^\tau_J$ commute, $\sigma$ is indeed a connection one-form.

By the construction,
$$
\frac{i}{2\pi}\left(\int_{\alpha_0} \sigma - \int_{\alpha_1} \sigma \right) = -\frac{im}{2\pi} \int_{b_{f_0}} \sigma  =  m.
$$
Since $\alpha_0 \sqcup \alpha_1$ bounds a cylinder $C \subset F^{-1}(\gamma \setminus f_0),$
we also have
$$
m = \dfrac{i}{2\pi} \int_{C} d\sigma  = \int_{E_\gamma / \mathbb S^{1}} {\bf c}_1,
$$
where ${\bf c}_1$ is the Chern class of the circle bundle $\rho \colon E_\gamma \to E_\gamma / \mathbb S^{1}$.
The result follows.
\end{proof}

\section{Discussion} \label{sec/discussion}

In this paper we studied Hamiltonian monodromy in integrable two-degree of freedom Hamiltonian systems with a circle action. We showed how Takens's index theorem, which is based on Morse theory, can be used to 
compute Hamiltonian monodromy. In particular, we gave a new proof of the monodromy around a focus-focus singularity using the Morse theory approach. 
An important implication of our results is a connection of the geometric theory developed in the works \cite{Efstathiou2017, Martynchuk2017} to Cushman's argument,
which is also based on Morse theory. New connections to the rotation number and to Duistermaat-Heckman theory were also discussed.

\section{Acknowledgements}
We would like to thank Prof. A. Bolsinov and Prof. H. Waalkens for useful and stimulating discussions.
We would also like to  thank the anonymous referee for his suggestions for improvement.

\appendix

\section{Hamiltonian monodromy} \label{sec/Hamiltonian_monodromy}

A typical situation in which monodromy arises is the case of an integrable system on a $4$-dimensional symplectic manifold $(M^4, \Omega)$. Such a system is specified by the \textit{energy-momentum} (or the \textit{integral}) map
\begin{equation*}
F = (H, J) \colon M \to \mathbb R^2.
\end{equation*}
Here $H$ is the Hamiltonian of the system and the momentum $J$ is a `symmetry' function, that is, the Poisson bracket 
$$\{H, J\} = \Omega^{-1}(dJ, dH) = 0$$ vanishes. We will assume that the map $F$ is proper, that is, that preimages of compact sets are compact, and that the fibers
$F^{-1}(f)$ of 
$F$ are connected. Then 
near any regular value of $F$ the functions $H$ and $J$ can be combined into new functions $I_1 = I_1(H,J)$ and 
$I_2 = I_2(H,J)$ such that the symplectic form has the canonical form 
$$
\Omega = dI_1 \wedge d\varphi_1 + dI_2 \wedge d\varphi_2
$$
for some angle coordinates $\varphi_1, \varphi_2$ on the fibers of $F$. This follows from the Arnol'd-Liouville theorem \cite{Arnold1968}. We note that
the regular fibers of $F$ are tori and that the motion on these tori is quasi-periodic.

The coordinates $I_i$ that appear in the Arnol'd-Liouville theorem are called \textit{action coordinates}. It can be shown that if $pdq$ is a local primitive $1$-from of the symplectic form,
then these coordinates are given by the formula
\begin{equation} \label{eq/actions}
I_i = \int\limits_{\alpha_i} p dq,
\end{equation}
where $\alpha_i, \ i = 1,2,$ are two independent cycles on an Arnol'd-Liouville torus. 
However, this formula is local even if the symplectic form $\Omega$ is exact. The reason for this is that the cycles
$\alpha_i$ can not, generally speaking, be chosen for each torus $F^{-1}(f)$ in a such a way that the maps $f \mapsto \alpha_i(f)$ are continuous at all regular values $f$  of $F$. 
This is the essence of \textit{Hamiltonian monodromy}. Specifically,
it is defined as follows.

Let $R \subset \textup{image}(F)$ be the set of the regular values of $F$. Consider the restriction map $$
F \colon F^{-1}(R) \to R.
$$
We observe that this map is a torus bundle: locally it is a direct product $D^n \times T^n$, the trivialization being achieved by the action-angle coordinates.
Hamiltonian monodromy is  defined as a representation
\begin{align*}
  \pi_1(R,f_0) \to \textup{Aut}\,H_1(F^{-1}(f_0))
\end{align*}
of the fundamental group $\pi_1(R,f_0)$ in the group of automorphisms of the integer homology group $H_1(F^{-1}(f_0))$. 
Each element $\gamma \in \pi_1(R,f_0)$ acts via parallel transport of integer homology cycles $\alpha_i$;
see \cite{Duistermaat1980}. 

We note that the appearance of the homology groups is due to the fact that the 
action coordinates \eqref{eq/actions} depend only on the homology class of $\alpha_i$ on the Arnol'd-Liouville torus. We observe that since the fibers of $F$ are tori, 
the group $H_1(F^{-1}(f_0))$ is isomorphic to $\mathbb Z^2$. It follows that the monodromy along a given path $\gamma$ is characterized by an integer matrix
$
M_\gamma \in \textup{GL}(2,\mathbb Z),
$
called the \textit{monodromy matrix} along $\gamma$. It can be shown that the determinant of this matrix equals $1$.

\begin{remark} \textit{(Examples and generalizations)}
Non-trivial monodromy has been observed in various examples of integrable systems, including the most fundamental 
ones, such as
the spherical pendulum \cite{Duistermaat1980, Cushman2015}, the hydrogen atom in crossed fields \cite{Cushman2000} and the spatial Kepler problem  \cite{DullinWaalkens2018, Martynchuk2019}. This invariant has also been generalized
in several different directions, leading to the notions of \textit{quantum} \cite{Cushman1988, Vu-Ngoc1999}, \textit{fractional}  
\cite{Nekhoroshev2006, Efstathiou2013, Martynchuk2017} and \textit{scattering} \cite{Bates2007, DullinWaalkens2008, Efstathiou2017, Martynchuk2019} \textit{monodromy}.
\end{remark}

\begin{remark} \textit{(Topological definition of monodromy)} Topologically, one can define
Hamiltonian monodromy along a loop $\gamma$ as monodromy of the torus (in the non-compact case --- cylinder) bundle  over this loop. More precisely, consider a
$T^2$-torus bundle
\begin{equation*}
F \colon F^{-1}(\gamma) \to \gamma, \ \gamma = S^1.
\end{equation*}
It can be obtained from a trivial bundle $[0,2\pi] \times T^2$ by gluing the boundary tori via a homeomorphism $f$, called the monodromy of $F$. 
In the context of integrable systems (when $F$ is the energy-momentum \textit{map} and $\gamma$ is a loop in the set of the regular values) the matrix of the push-forward map 
$$f_{\star} \colon \textup{H}_1(T^2) \to \textup{H}_1(T^2)$$
coincides with the monodromy matrix along $\gamma$ in the above sense. 
It follows, in particular, that monodromy can be defined for any torus bundle.
\end{remark}

\section{Chern classes} \label{appendix/chern}
Let $M'$ be an $\mathbb S^1$-invariant submanifold of $M$ which does not contain the critical points of $H$. The circle action on $M'$ is then free and we have a
principal circle bundle
\begin{equation*}
 \rho \colon M' \to M'/ \mathbb S^1.
\end{equation*}
Let $X_J$ denote the vector field on $M'$ corresponding to the circle action (such that the flow of $X_J$ gives the circle action) and let $\sigma$ be a $1$-form on $M'$ such that the following two conditions hold

\begin{center}
(i) $\sigma(X_J) = i$ and  (ii) $R^{*}_g(\sigma) = \sigma.$
\end{center}

Here $i \in i\mathbb R$ --- the Lie algebra of $\mathbb S^1 = \{e^{i\varphi} \in \mathbb C \mid \varphi \in [0,2\pi]\}$ and $R_g$ is the (right) action of $\mathbb S^1$.

The Chern (or the Euler) class\footnote{this Chern class should not be confused with Duistermaat's Chern class, which is another obstruction to the existence of global action-angle coordinates; see
\cite{Duistermaat1980, Lukina2008}.}can then defined as 
$${\bf c}_1 = s^{\star}(i dw / 2\pi) \in H^2(M'/ \mathbb S^1,\mathbb R),$$ where $s$ is any local section of the circle bundle $\rho \colon M' \to M'/ \mathbb S^1$.
Here $H^2(M'/ \mathbb S^1,\mathbb R)$ stands for the second de Rham cohomology group of the quotient $M'/ \mathbb S^1$. 

We note that if the manifold
$M'$ is compact and $3$-dimensional, the Chern number of $M'$ (see Definition~\ref{def/Chern_number}) is equal to the integral
\begin{equation*}
 \int_{M'/ \mathbb S^1} {\bf c}_1
\end{equation*}
of the Chern class ${\bf c}_1$ over the base manifold $M'/ \mathbb S^1$.

A non-trivial example of a circle bundle with non-trivial Chern class is given by the (anti-)Hopf fibration. Recall that
the Hopf fibration of the $3$-sphere
$$S^3 = \{(z,w) \in \mathbb C^2 \mid 1 = |z|^2+|w|^2\}$$
is the principal circle bundle $S^3 \to S^2$ obtained by reducing the circle action 
$(z,w) \mapsto (e^{it}z,e^{it}w)$. 
  The circle action $(z,w) \mapsto (e^{-it}z,e^{it}w)$ defines the anti-Hopf fibration of $S^3$.

\begin{lemma}
 The Chern number of the Hopf fibration is equal to $-1$, while for the anti-Hopf fibration it is equal to $1$.
\end{lemma}
\begin{proof}
 Consider
  the case of the Hopf fibration  (the anti-Hopf case is analogous). Its projection map $h \colon S^3 \to S^2$ is defined by 
  $h(z,w) = (z:w) \in \mathbb {CP}^1 = S^2.$
  Put 
  $$U_1 = \{(u:1) \mid u \in \mathbb C, \ |u| < 1\} \mbox{ and }  U_2 = \{(1:v) \mid v \in \mathbb C, \ |v| < 1\}.$$ 
  Define the section $s_j \colon U_j \to S^3$ by the formulas
 $$s_1((u:1)) = \left(\dfrac{u}{\sqrt{|u|^2+1}},\dfrac{1}{\sqrt{|u|^2+1}}\right)$$
 and
 $$s_2((1:v)) = \left(\dfrac{1}{\sqrt{|v|^2+1}},\dfrac{v}{\sqrt{|v|^2+1}}\right). $$
 Now, the gluing cocycle $t_{12} \colon S^1 = \overline{U}_1 \cap \overline{U}_2 \to \mathbb S^1$ corresponding to the sections $s_1$ and $s_2$ is given by
\begin{equation*}
t_{12}((u:1)) = \exp{ (- i \textup{Arg} \, u)}.
\end{equation*}
If follows that the winding number equals $-1$ (the loop $\alpha$ in  Definition~\ref{def/Chern_number}
is given by the equator $S^1 = \overline{U}_1 \cap \overline{U}_2$ in this case).
\end{proof}

\bibliographystyle{amsplain}
\bibliography{library}

\end{document}